\documentclass[aps,pra,10pt,twocolumn,showpacs,superscriptaddress,footnoteinbib]{revtex4}
\usepackage{amsmath,float,amsthm,verbatim}
\usepackage{latexsym}
\usepackage{amssymb,dsfont}
\usepackage{color}
\usepackage{mathtools}
\usepackage{graphicx}
\usepackage[outdir=./]{epstopdf}
\usepackage{graphics}
\usepackage{soul}
\usepackage[breaklinks]{hyperref}
\newtheorem{prop}{Proposition}

\begin{document}

\title{No nonlocal advantage of quantum coherence beyond quantum instrumentality}
\author{Debasis Mondal}
\thanks{cqtdem@nus.edu.sg} 
\affiliation{Centre for Quantum Technologies, National University of Singapore, 3 Science Drive 2, Singapore 117543}
\author{Jaskaran Singh}
\thanks{jaskaransinghnirankari@iisermohali.ac.in}
\affiliation{Department of Physical Sciences, Indian Institute of Science Education and Research (IISER) Mohali,
	Sector 81 SAS Nagar, Manauli PO 140306 Punjab India.}
\author{Dagomir Kaszlikowski}
\thanks{phykd@nus.edu.sg} 
\affiliation{Centre for Quantum Technologies, National University of Singapore, 3 Science Drive 2, Singapore 117543}
\affiliation{Department of Physics, National University of Singapore, 2 Science Drive 3, 117542 Singapore, Singapore}
\date{\today}


\begin{abstract}
	Recently, it was shown that quantum steerability is stronger than the bound set by the instrumental causal network. This implies, quantum instrumentality cannot simulate EPR-nonlocal correlations completely. In contrast, here we show that quantum instrumentality can indeed simulate EPR correlations completely and uniquely if viewed from the perspective of NAQC. Implication of our result is that the entire set of EPR-correlations can be explained by the LHS model in the instrumental causal background if viewed from the perspective of NAQC.
\end{abstract}


\maketitle
\section{ Introduction.}
Our fundamental interest in science is to debunk the mystery of our universe by revealing its underlying laws. Quantum mechanics, so far has been the most successful theory of our universe. However, basis of why quantum mechanics works still puzzles us. There have been several attempts to simulate quantum mechanical results using the principle of locality, non-contextuality, determinism or realism and free will to name a few. In an attempt to provide an ontological model for quantum mechanics, a number of no-go theorems including the Bell theorem, Bell-Kochen-Specker theorem and the EPR theorem \cite{jsbell,lhsintro,kochenspecker} were proposed using such theory-independent, physically motivated principles with the hope to single out quantum theory as a primary model to explain nature among the plethora of generalized probabilistic models. 

Failure of these theory-independent no-go theorems led to an important question as to why nature does not allow stronger correlations than what quantum mechanics permits \cite{dam,brunner,popescu2006,brassad}. A number of theory-dependent no-go theorems were also discovered including the no cloning theorem, no deleting theorem etc. In \cite{dam}, it was shown that any post-quantum theory exhibiting stronger non-local correlations may lead to a `computational free lunch', which enables all distributed computations with a trivial amount of communication, i.e. with one bit. So far, there has been a lack of any physical principle, which could prevent such `computational free lunch' and at the same time, also identify quantum theory uniquely. Many partially successful attempts have already been made in this direction proposing several new principles like information causality \cite{pawlowski}, 
local orthogonality \cite{fritz,navascues}, exclusivity \cite{acabelo}, no-causal order \cite{fcosta,deok}, non-trivial communication complexity \cite{dam,brassad} and macroscopic locality \cite{navascues1}. Even though there has been significant progress, a no-go theorem based on a set of physically motivated laws or principles, which could simulate quantum mechanical results uniquely, is still unknown. 

 Recently, ontological models based on the principle of causality have garnered interests and led to the development of the field of quantum causal modeling \cite{fcosta1,chiribella1,chaves1,pienaar,mark,portmann,causallyneutral,henson},  which brought forth a number of fascinating results including the idea of no definite global causal order \cite{fcosta,deok} and its applications \cite{chiribelacomp}. Surprisingly, it has also been experimentally verified recently \cite{rubino}. Therefore, it is now natural to ask  whether the violation of various no-go theorems is because of this stricter notion of ordered causal relation.

To that end, several attempts have already been made to understand quantum nonlocality \cite{chavesbell1,chavesbell2,chavesbell3,rossetbell1,wolfebell1,fritzbell1,fritzbell2}, contextuality \cite{sally} and EPR-nonlocality \cite{nerytaddei} relaxing the stricter notion of ordered causal relation. In this regard, the quantum instrumental causal network is one of the most promising structure of causal models. Quantum instrumental processes are a generalization of their classical counterparts with quantized communication receiving nodes with underlying local hidden variable (LHV) or state (LHS) model and outcome communications (see Fig. \ref{fig1}). 

Instead of the traditional ordered causal network, search for various new quantum causal models has drawn quite a bit of attention \cite{fcosta2,nerytaddei,wolfebell2,chavesbell4} in the last few years and the model, `instrumental causal network' turns out to be the most prominent candidate in this regard. It may be considered as a relaxed version of our day-to-day observation of cause and effect relationship. In the derivation of all the no-go theorems, we assume to live in a world with ordered causal background.  This provides us the opportunity to relax the idea of causal relationship from the existing no-go theorems and swim closer to the direction of singling out the entire set of quantum correlations. 

So far, relaxing the prevalent idea of causal network has not been beneficial in this regard. In fact, in a recent work, device independent instrumental inequality was shown to admit a quantum violation \cite{bonet, chaves33}. On the other hand, EPR-nonlocality also was shown to be stronger than the bound set by the instrumental causal network \cite{nerytaddei} or one sided quantum instrumental network (1SQI).  

In this paper, in an attempt to find a theory-dependent principle or no-go theorem of quantum mechanics, we study steerability of a state 
from the perspective of non-local advantage of quantum coherence (NAQC) \cite{deba1, deba4, deba5}. We derive a set of new tighter steering inequalities based on various coherence measures in the ordered causal background. Violation of these inequalities implies nonlocal advantage of quantum coherence (NAQC) beyond what a single system can achieve. We then derive a similar set of inequalities under 1SQI model \cite{nerytaddei} or in an instrumental causal background. It turns out that unlike quantum steering viewed from the perspective of entropy or uncertainty, NAQC in the ordered causal background is upper bounded by the 1SQI bound. 

Implication of our result is that the entire set of EPR-correlations can be described by the LHS model in the instrumental causal background if viewed from the perspective of NAQC.

 \begin{figure}
\includegraphics[scale=.45]{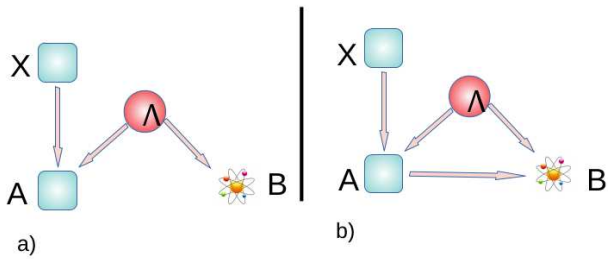}
\caption{The graphical representation of the LHS model and the 1SQI is shown using DAGs. Each node encodes either a classical random variable or a quantum system, and each directed edge a causal influence. Circular nodes denote unobservable variables and square shaped nodes and graphic image of atomic structural nodes denote observable variables. The first type nodes represent classical observables and the second type nodes for quantum observables. It is clear from the presence of the directed edge between Alice ($A$) and Bob ($B$) in the second DAG (b) that the second DAG represents 1SQI and the other one on the left represents LHS model.}
\label{fig1}
\end{figure}

\section{One sided quantum instrumentality(1SQI)}We consider a steering scenario, where Alice prepares a bipartite state and sends a part of the system to Bob, who does not trust her. Alice tries to convince Bob that his state is entangled with hers. To that end, Bob asks Alice to perform certain tasks. Bob believes that there exists an unobservable shared source $\Lambda$ influencing both of them. In local hidden state (LHS) model, Bob thinks that there is no direct causal influence from Alice to Bob but here we relax that assumption and consider that there is indeed a direct causal influence from Alice to Bob as shown in Fig.(\ref{fig1}) b) through a directed acyclic graphs (DAGs). The conditional states of Bob $\rho_{a|x}$ (unnormalized) under 1SQI model are then represented by
\begin{equation}
\rho_{a|x}=\sum_{\lambda}P_{\lambda}p(a|x,\lambda)\rho_{\lambda,a},
\label{1sqieq}
\end{equation}
where $P_{\lambda}$ is a probability distribution over the hidden variables $\lambda$ assigned to the node $\Lambda$, $p(a|x, \lambda)$ is the conditional probability of obtaining outcome $a$ for the measurement setting $x$ and hidden variable $\lambda$ to node $A$ and $\rho_{\lambda,a}$ is the LHS with $\text{Tr}(\rho_{\lambda,a})=1$ assigned to node $B$ by the model. In contrast, the conditional states under the usual LHS model have a representation as
\begin{equation}
\rho_{a|x}=\sum_{\lambda}P_{\lambda}p(a|x,\lambda)\rho_{\lambda}.
\label{lhseq}
\end{equation}

\section{Coherence complementarity relation.}This section is dedicated to the derivation of the non-local advantage of local quantum coherence under the usual ordered causal relation as well as instrumental causal relation. To start with, we first derive the coherence complementarity inequalities based on various measures of quantum coherence for a single qubit state. We consider a general qubit state $\rho=\frac{1}{2}(I_{2}+\vec{r}.\vec{\sigma})$, where $|\vec{r}|\leq 1$ and $\vec{\sigma}\equiv(\sigma_1, \sigma_2, \sigma_3)$ are the Pauli matrices. The coherence of the state when expressed in the eigenbasis of $\sigma_i$, can be expressed by the $l_1$-norm of coherence $(C^{l_1})$ as
\begin{equation}
C_{i}^{l_1}=\sqrt{r_{j}^2+r_k^{2}},
\label{eql1}
\end{equation}
where $i\neq j\neq k$. For the remainder of the article we adopt the notation $i, j, k\in \{1, 2, 3\} $.

Similarly, the relative entropy of coherence $(C^{r})$ with respect to the $i^{th}$ basis is given by
\begin{equation}
C_{i}^r=\mathcal{H}\left(\frac{1+r_i}{2}\right)-3\mathcal{H}\left(\frac{1+|\vec{r}|}{2}\right),
\label{eqr}
\end{equation}
where $\mathcal{H}(x)=-x\log_{2}x-(1-x)\log_{2}(1-x)$ for $0\leq x\leq 1$. Using the fact that $(C_{1}^{\alpha}-C_{2}^{\alpha})^2+(C_{2}^{\alpha}-C_{3}^{\alpha})^2+(C_{3}^{\alpha}-C_{1}^{\alpha})^2\geq 0$, where we take $\alpha\in \{l_1, r\}$, the expressions in Eq.~(\ref{eql1}) and Eq.~(\ref{eqr}) reduce to
\begin{equation}
C_{1}^{\alpha}C_2^{\alpha}+C_2^{\alpha}C_3^{\alpha}+C_3^{\alpha}C_1^{\alpha} \leq \sum_{i=1}^{3}C_i^{\alpha^2}\leq \Omega_\alpha,
\label{cohcomp}
\end{equation}
where $\Omega_{l_1}=\Omega_{r}=2$ for an arbitrary qubit state. Using a similar approach it is possible to extend the inequality for arbitrary dimension (see the supplemental material \cite{supple}).

\section{Non-local advantage of quantum coherence.} We now derive a new NAQC inequality under both ordered and instrumental causal networks. Without loss of generality and for simplicity, we limit our analysis within the regime of two-qubit states, while the results can be easily extended for any general bipartite states (see supplemental material \cite{supple} for general proof). We consider a two qubit bipartite state $\rho_{ab}$ prepared by Alice and shared with Bob. We also assume that the conditional states of Bob admit an LHS or 1SQI model as given by Eq.~(\ref{lhseq}) and Eq. (\ref{1sqieq}) respectively. If Bob measures certain properties of his states such as entropy, uncertainty, it has been shown to violate the steering inequalities \cite{ricardo,anandam, sourodip, walborn}. It has also been shown to violate the 1SQI inequality based on semi-definite programming (SDP) \cite{nerytaddei}. The quantities or the properties considered (uncertainty and entropy) so far has a classical counterpart. Both the uncertainty as well as the entropy of the state of Bob has contribution from classical mixedness of the state. Classical mixing and thermal noise directly contribute to such quantities and thus, affects the non-locality. On the other hand, quantum coherence is the absence of classical mixing and thermal noise \cite{uttamsingh} and thus, noise plays no role in the non-locality measured based on a coherence dependent quantity (less robust under noise). A steering inequality based on quantum coherence thus naturally provides a different view of the situation. It also depicts how a purely quantum resource (quantum correlation) affects another quantum resource (coherence). In the following, we show that unlike the traditional inequalities based on uncertainty or entropy, a 1SQI inequality based on quantum coherence can indeed single out the EPR-correlation viewed via NAQC. We start with the sum of square of average local quantum coherence of Bob's state in the mutually unbiased bases, i.e. for a two-qubit scenario,
\begin{equation}
S:=\sum^1_{a,b=0}\sum_{i\neq j\neq k}p(a|i)C^{\alpha}_{k}(\rho'_{a|i})p(b|j)C^{\alpha}_{k}(\rho'_{b|j}),
\label{quantity}
\end{equation}
where $\rho'_{c|x}=\frac{\rho_{c|x}}{\text{Tr}(\rho_{c|x})}=\frac{\rho_{c|x}}{\sum_{\lambda}P_{\lambda}p(c|x,\lambda)}$ is the normalized conditional state of Bob and $p(c|x)=\text{Tr}(\rho_{c|x})=\sum_{\lambda}P_{\lambda}p(c|x, \lambda)$ is the probability of being in the state. As we next show, the quantity $S$ has a nontrivial bound under both the ordered and instrumental causal network. We derive bounds for both the cases below.
\begin{prop}
Under LHS model and ordered causal network, the quantity $S$ that Bob measures on his particle is bounded as
\begin{equation}
S\overset{\text{LHS}}{\leq} 2\Omega_{\alpha}.
\label{ineq2}
  \end{equation}
  \end{prop}

\begin{proof}
	A proof of the above under a LHS model is outlined  below.

\begin{widetext}
\begin{eqnarray}
S&=&\sum^1_{a,b=0}\sum_{i\neq j\neq k}p(a|i)p(b|i)C^{\alpha}_{k}(\rho'_{a|i})C^{\alpha}_{k}(\rho'_{b|j})\overset{LHS}{\leq} \frac{1}{2}\underset{\lambda,\lambda'}{\sum^1_{a,b=0}\sum_{i\neq j\neq k}} P_{\lambda}P_{\lambda'}p(a|i,\lambda)p(b|j,\lambda')\left(C^{\alpha^2}_k(\rho'_{\lambda})+C^{\alpha^2}_k(\rho'_{\lambda'})\right)\nonumber\\&=&\frac{1}{2}\underset{\lambda,\lambda'}{\sum^1_{a,b=0}\sum_{k}}P_{\lambda}P_{\lambda'}(p(a|\mathcal{K}_1^k,\lambda)p(b|\mathcal{K}_2^k,\lambda')+p(b|\mathcal{K}_1^k,\lambda')p(a|\mathcal{K}_2^k,\lambda))\left(C^{\alpha^2}_k(\rho'_{\lambda})+C^{\alpha^2}_k(\rho'_{\lambda'})\right)\nonumber\\&\leq &
\underset{\lambda,\lambda'}{\sum_{k}}P_{\lambda}P_{\lambda'}\left(C^{\alpha^2}_k(\rho'_{\lambda})+C^{\alpha^2}_k(\rho'_{\lambda'})\right)
\leq 2\sum_{\lambda,\lambda'}P_{\lambda}P_{\lambda'}\Omega_{\alpha}
= 2\Omega_{\alpha}
\label{proof5}
\end{eqnarray}
\end{widetext}

where $\mathcal{K}_n^k=\text{Mod}(k-1+n,3)+1$. The first inequality in Eq.~(\ref{proof5}) comes from the fact that conditional states have representations as given by Eq.~(\ref{lhseq}) and the fact that coherence does not increase under classical mixing of states, i.e., for a state $\rho=\sum_{i}p_{i}\rho_{i}$ such that $\sum_{i}p_{i}=1$, $C^{\alpha}(\rho)
\leq\sum_{i}p_{i}C^{\alpha}(\rho_{i})$. We have also used the fact that for any real numbers $x$ and $y$, $xy\leq\frac{x^2+y^2}{2}$. 


\end{proof}

\begin{figure}
	\includegraphics[scale=1]{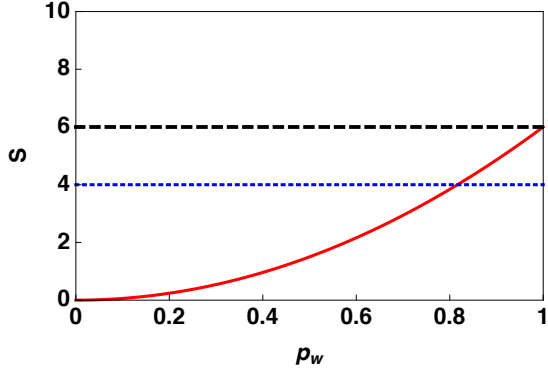}
	\caption{We plot $S$ (solid, red) with the varying values of $p_w$ for the Werner state~(\ref{werner}) using the $l_1$-norm measure of quantum coherence. We use Pauli measurements in arbitrary directions and optimize over directions ($\theta$ and $\phi$) for maximum value of $S$ as detailed in \cite{supple}.
	$S$ violates the bound (dotted, blue) given by the inequality in Eq.~(\ref{ineq2}) but does not violate the bound (dashed, black) given by the inequality in Eq.~(\ref{ineq1}). }
	\label{s2}
\end{figure}


Violation of the above inequality for any quantum state not only implies that the state is steerable but also shows that Bob can achieve the nonlocal advantage of quantum coherence beyond what could have been possible without the intervention of Alice nonlocally. In \cite{deba4}, a set of steering complementarity relations were derived. Here we show a set of similar complementarity relations in the supplemental material \cite{supple}. Moreover, we also show that the inequality in Eq. (\ref{ineq2}) is tight, i.e., there exist a state with LHS model ($\rho_{ab}=|0\rangle\langle 0|\otimes|+\rangle\langle +|$  for example), which  can achieve the bound. 



In the next section, we focus on deriving a similar bound on the quantity $S$ under the 1SQI model with outcome communications. 
\begin{prop}
	If Bob assumes that his conditional states admit descriptions as given by 1SQI model in Eq.~(\ref{1sqieq}) and measures the quantity $S$ on his states, it must be bounded by
	\begin{equation}
	S\overset{1SQI}{\leq} 3\Omega_{\alpha}.
	\label{ineq1}
	\end{equation}
\end{prop}

\begin{proof}
Proof of this inequality follows along the same line of approach as before.

\begin{eqnarray}
	S&=&\sum^1_{a,b=0}\sum_{i\neq j\neq k}p(a|i)p(b|j)C^{\alpha}_{k}(\rho'_{a|i})C^{\alpha}_{k}(\rho'_{b|j})\nonumber\\&\overset{1SQI}{\leq} &\underset{\lambda, \lambda'}{\sum^1_{a,b=0}\sum_{i\neq j\neq k}}\begin{aligned}P_{\lambda}p(a|i, \lambda)&P_{\lambda'}p(b|j, \lambda')\\&C^{\alpha}_{k}(\rho_{a,\lambda})C^{\alpha}_{k}(\rho_{b,\lambda'})\end{aligned}\nonumber\\
	&\leq &\frac{1}{2}\underset{\lambda, \lambda'}{\sum_{i\neq k}}P_\lambda P_{\lambda'} \left(\begin{aligned} &\sum_{a, b=0}^{1}\bigg{(}p(a|i,\lambda)C^{\alpha}_{k}(\rho_{b,\lambda'})\bigg{)}^2\\& + \sum_{a, b=0}^{1}\bigg{(}p(b|i,\lambda')C^{\alpha}_{k}(\rho_{a,\lambda'})\bigg{)}^2\end{aligned}\right)\nonumber\\&=&\frac{1}{2}\underset{\lambda,\lambda'}{\sum^1_{a,b=0}\sum_{k}}P_{\lambda}P_{\lambda'}\left(\begin{aligned} &F_k(a|\lambda)C^{\alpha^2}_k(\rho'_{b, \lambda'})\\&+F_k(b|\lambda')C^{\alpha^2}_k(\rho'_{a, \lambda})\end{aligned}\right),
\label{proof1}
\end{eqnarray}

\begin{figure}
	\includegraphics[scale=1]{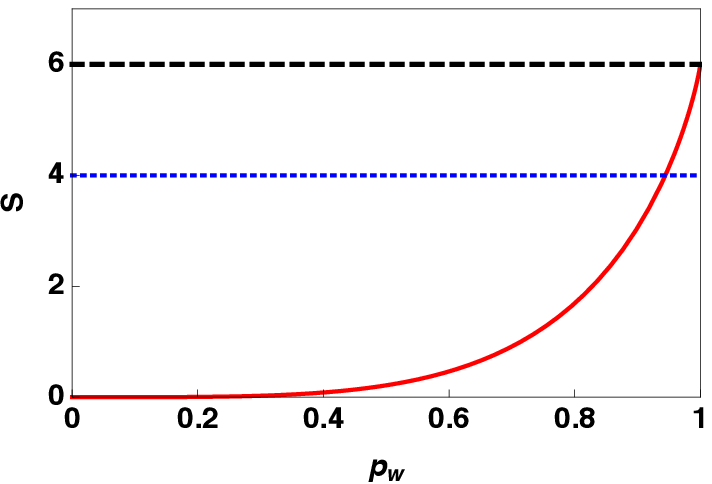}
	\caption{We plot $S$ (solid, red) with the varying values of $p_w$ for the Werner state~(\ref{werner}) using the relative entropy measure of quantum coherence. We use Pauli measurements in arbitrary directions and optimize over directions ($\theta$ and $\phi$) for maximum value of $S$ \cite{supple}. $S$ violates the bound (dotted, blue) given by the inequality in Eq. (\ref{ineq2})   but does not violate the bound (dashed, black) given by the inequality in Eq. (\ref{ineq1}).}
	\label{s1}
\end{figure}

where $F_k(a|\lambda) = p(a|\mathcal{K}_1^k,\lambda)^2+p(a|\mathcal{K}_2^k,\lambda)^2$. As before, it can be shown \cite{supple} that for an arbitrary qubit state,
\begin{equation}
\sum_{a=0}^{1}F_k(a|\lambda)\leq \frac{3}{2} \quad \forall k.
\label{f}
\end{equation}
Plugging Eqn.~(\ref{f}) in (\ref{proof1}), we get,
\begin{eqnarray}
S&=&\sum^1_{a,b=0}\sum_{i\neq j\neq k}p(a|i)p(b|j)C^{\alpha}_{k}(\rho'_{a|i})C^{\alpha}_{k}(\rho'_{b|j})\nonumber\\&\leq& 3\Omega_{\alpha}
\label{1sqibound}
\end{eqnarray}
As before, in the first inequality in Eq. (\ref{proof1}), we use the 1SQI model as given in Eq.~(\ref{1sqieq}) and the fact that coherence does not increase under classical mixing. The second inequality is a consequence of the fact that for any two real numbers $x$ and $y$, $xy\leq\frac{x^2+y^2}{2}$ (see the generalization of the bound in section $F$ of \cite{supple}). In the last inequality in Eq. (\ref{1sqibound}), we use the coherence complementarity relation as given in Eq. (\ref{cohcomp}). A generalized form of coherence complementarity relationship for states in arbitrary dimensions can be used to generalize the above proof for two qudits. Furthermore, in the supplementary material~\cite{supple}, we show that the above inequality is also tight.
\end{proof}

For example for both the cases of ordered and instrumental causal networks, we consider Werner states, given by 
\begin{equation}
\rho_{ab}=\frac{1-p_w}{4}\mathds{1}_4+p_w|\psi_{ab}\rangle\langle\psi_{ab}|,
\label{werner}
\end{equation}
where $0\leq p_w \leq 1$ and $|\psi_{ab}\rangle$ is the Bell singlet state. We plot the behavior of $S$ with respect to $p_w$ in Fig.~\ref{s2} and Fig.~\ref{s1} respectively using both $l_1$-norm and relative entropy of coherence as measures of coherence. From the plots, we find that the quantity $S$ violates the bound set by the ordered causal network for $p_w>0.816$ for the $l_1$-norm of coherence and $p_w>0.944$ for the relative entropy of coherence.

On the other hand,  the inequality in Eq. (\ref{ineq1}) is not violated by the Werner state for any range of $p_w$ for the $l_1$-norm measure of quantum coherence in Fig. (\ref{s2}) as well as relative entropy measure of quantum coherence in Fig. (\ref{s1}). One can in fact show that 

\begin{prop}
no two-qubit state can violate the bound set by the 1SQI inequality as given in Eq. (\ref{ineq1}), i.e.,
\begin{equation}
\max_{\rho_{ab}\in \mathcal{L}_{4}^{+}} S\leq 3\Omega_{\alpha}.
\label{ineq2}
  \end{equation}
  \end{prop}
\begin{proof}
\begin{eqnarray}\label{eq:snew}
S&=&\sum_{a,b=0}^{1}\sum_{i\neq j\neq k}p(a|i)p(b|j)C_k^\alpha(\rho'_{a|i})C_k^\alpha(\rho'_{b|j})\nonumber\\
&\leq & \frac{1}{2}\sum_{a,b=0}^{1}\sum_{i\neq j\neq k}p(a|i)p(b|j)(C_k^{\alpha^2}(\rho'_{a|i})+C_k^{\alpha^2}(\rho'_{b|j}))\nonumber\\
&=&\frac{1}{2}\sum_{i\neq j\neq k}\left(\sum_{a}p(a|i)C_k^{\alpha^2}(\rho'_{a|i})+\sum_{b}p(b|j)C_k^{\alpha^2}(\rho'_{b|j})\right)\nonumber\\
&= & \frac{1}{2}\sum_{i\neq j\neq k}\left(S^{(1)}_{a,i,k}+S^{(2)}_{b,j,k}\right),
\end{eqnarray}
\end{proof}
where $S^{(1)}_{a,i,k} = \sum_a p(a|i)C_k^{\alpha^2}(\rho'_{a|i})$ and $S^{(2)}_{b,j,k} = \sum_b p(b|j)C_k^{\alpha^2}(\rho'_{b|j})$.  Eq.~(\ref{eq:snew}) can be evaluated term wise as,
\begin{eqnarray}
\sum_{\i \neq j\neq k}S^{(1)}_{a,i,k} &=& \sum_{a,i\neq j\neq k}p(a|i)C_k^{\alpha^2}(\rho'_{a|i})\nonumber\\
&=&\sum_{i\neq j \neq k}p(0|i)C_k^{\alpha^2}(\rho'_{0|i})+p(1|i)C_k^{\alpha^2}(\rho'_{1|i})\nonumber\\
&\leq & 6,
\end{eqnarray}
where the last inequality is due to the fact the maximum value of coherence is one and for any probability distribution $p(x)$ and any positive function $f(x)$, $\sum_{x}p(x)f(x)\leq\sum_{x}f(x)$. After a similar analysis for $\sum_{i\neq j\neq k}S^{(2)}_{b,j,k}$, we get
\begin{equation}
S\leq 6,
\end{equation}
which concludes the proof that no two qubit state can violate the bound $3\Omega_{\alpha}$. 

This can again be generalized to arbitrary two-qudit states by appropriately choosing the generalized coherence complementarity relationship~(\ref{cohcomp}) (see the supplemental material \cite{supple}). We come to the same conclusion even after studying the general two-qudit state as shown in the supplemental material \citep{supple}.

\section{ Conclusions and discussions.}It was shown that there exist quantum states, which exibit stronger EPR-nonlocality than the bound set by 1SQI (\cite{nerytaddei}), i.e., quantum steering cannot be explained by 1SQI model. However, in this article, we start with a new and stronger steering inequality based on local quantum coherence \cite{deba1, deba4} under the ordered causal network. We observe violation of the inequality as shown in Fig. (\ref{s2}) and (\ref{s1}) and term the phenomena as NAQC. We derive a similar bound on the quantity under instrumental causal network and outcome communication. Like \cite{nerytaddei}, a violation of the inequality in Eq. (\ref{ineq1}) certifies that NAQC beyond quantum instrumentality is possible just like quantum steerability. However, we observe that although quantum steering in general can be more stronger than what quantum instrumentality allows, NAQC is in turn upper bounded by the 1SQI bound. The new inequality for instrumental causal network is never violated by the NAQC for any state. This implies that the entire set of EPR-correlations can be explained by the LHS model in the instrumental causal background if viewed from the perspective of NAQC.

There have been several attempts to single out quantum or more precisely, physical correlations based on different physical principles or no-go theorems.  In this paper, we show that although 1SQI bound cannot single out quantum correlations when viewed from the perspective of violation of local quantum analogue of classical properties, but in principle, can identify the correlations when viewed from the perspective of NAQC. While identification of physical correlations using different principles or no-go theorems is considered to be a non-trivial task, we believe, our efforts indeed advance us significantly in this particular direction.

Since we have shown that an inequality based on quantum coherence fits just perfectly for arbitrary dimensions, such that the 1SQI bound uniquely singles out NAQC correlations, one of the immediate questions naturally arises: Why does quantum coherence play such an important role?  For the sake of arguments, even if we consider that we stumbled upon a pair of two perfect quantities, namely EPR correlation and coherence, by a mere serendipitous coincidence, we cannot ignore that they are a perfect match for each other, which leads to another immediate question: what makes them the perfect match? We believe, quantum coherence does not play any role but rather it is transition probabilities of a state, which makes a perfect pair with the quantum correlations, wheres quantum coherence is just a function of these probabilities. Thus, in future, it will be really an important task to investigate the pairs more closely.

%
%

{\bf Acknowledgements---}D.M. would like to acknowledge the support from the
National Research Foundation. JS would like to acknowledge funding from UGC, India. D.K. is supported by the
National Research Foundation and the Ministry of Education in Singapore through the Tier 3 MOE2012-T3-1-
009 Grant Random numbers from quantum processes.

\newpage
\clearpage
\onecolumngrid
\section{Supplemental Material}
\newcounter{defcounter}
\setcounter{defcounter}{0}
\newtheorem{defn}{Definition}
\newcommand{\comb}[2]{{}^{#1}\mathrm{C}_{#2}}
\newenvironment{myequationA}{%
\addtocounter{equation}{-1}
\refstepcounter{defcounter}
\renewcommand\theequation{A\thedefcounter}
\begin{equation}}
{\end{equation}}
\newenvironment{myequationB}{%
\addtocounter{equation}{-1}
\refstepcounter{defcounter}
\renewcommand\theequation{B\thedefcounter}
\begin{equation}}
{\end{equation}}
\newenvironment{myequationC}{%
\addtocounter{equation}{-1}
\refstepcounter{defcounter}
\renewcommand\theequation{C\thedefcounter}
\begin{equation}}
{\end{equation}}
\newenvironment{myequationD}{%
\addtocounter{equation}{-1}
\refstepcounter{defcounter}
\renewcommand\theequation{D\thedefcounter}
\begin{equation}}
{\end{equation}}
\newenvironment{myequationE}{%
\addtocounter{equation}{-1}
\refstepcounter{defcounter}
\renewcommand\theequation{E\thedefcounter}
\begin{equation}}
{\end{equation}}
\newenvironment{myequationF}{%
\addtocounter{equation}{-1}
\refstepcounter{defcounter}
\renewcommand\theequation{F\thedefcounter}
\begin{equation}}
{\end{equation}}
\newenvironment{myequationG}{%
\addtocounter{equation}{-1}
\refstepcounter{defcounter}
\renewcommand\theequation{G\thedefcounter}
\begin{equation}}
{\end{equation}}
\newenvironment{myequationH}{%
\addtocounter{equation}{-1}
\refstepcounter{defcounter}
\renewcommand\theequation{H\thedefcounter}
\begin{equation}}
{\end{equation}}
\newenvironment{myeqnarrayA}{%
\addtocounter{equation}{-1}
\refstepcounter{defcounter}
\renewcommand\theequation{A\thedefcounter}
\begin{eqnarray}}
{\end{eqnarray}}
\newenvironment{myeqnarrayB}{%
\addtocounter{equation}{-1}
\refstepcounter{defcounter}
\renewcommand\theequation{B\thedefcounter}
\begin{eqnarray}}
{\end{eqnarray}}
\newenvironment{myeqnarrayC}{%
\addtocounter{equation}{-1}
\refstepcounter{defcounter}
\renewcommand\theequation{C\thedefcounter}
\begin{eqnarray}}
{\end{eqnarray}}
\newenvironment{myeqnarrayD}{%
\addtocounter{equation}{-1}
\refstepcounter{defcounter}
\renewcommand\theequation{D\thedefcounter}
\begin{eqnarray}}
{\end{eqnarray}}
\newenvironment{myeqnarrayE}{%
\addtocounter{equation}{-1}
\refstepcounter{defcounter}
\renewcommand\theequation{E\thedefcounter}
\begin{eqnarray}}
{\end{eqnarray}}
\newenvironment{myeqnarrayF}{%
\addtocounter{equation}{-1}
\refstepcounter{defcounter}
\renewcommand\theequation{F\thedefcounter}
\begin{eqnarray}}
{\end{eqnarray}}
\newenvironment{myeqnarrayG}{%
\addtocounter{equation}{-1}
\refstepcounter{defcounter}
\renewcommand\theequation{G\thedefcounter}
\begin{eqnarray}}
{\end{eqnarray}}
\newenvironment{myeqnarrayH}{%
\addtocounter{equation}{-1}
\refstepcounter{defcounter}
\renewcommand\theequation{H\thedefcounter}
\begin{eqnarray}}
{\end{eqnarray}}
\subsection{A. Pauli operators in arbitrary directions}Pauli operators are well known matrices in physics, generally denoted by $\sigma_x$, $\sigma_y$, $\sigma_z$ or $\sigma_1$, $\sigma_2$ and $\sigma_3$, where
\begin{widetext}
$\sigma_1=\begin{pmatrix}
0 & 1 \\
1 & 0
\end{pmatrix}, 
\sigma_2=\begin{pmatrix}
0 & -i \\
i & 0
\end{pmatrix},
\sigma_3=\begin{pmatrix}
1 & 0 \\
0 & -1
\end{pmatrix}
$.
\end{widetext}
Eigen bases of these operators turn out to be the mutually unbiased bases(MUBs). One can, in principle, rotate the operators in arbitrary directions. For our purpose, without loss of generality, we rotate the matrices such that the 
Eigen vectors of the rotated  $\sigma_3$ in the $(\theta, \phi)$ direction($
\sigma_{3}(\theta, \phi)$) turn out to be a set of arbitrary ortho-normal states 
such as $|0(\theta, \phi)\rangle=\cos(\frac{\theta}{2})|0\rangle+e^{i\phi}
\sin(\frac{\theta}{2})|1\rangle$ and $|1(\theta, \phi)\rangle=\sin(\frac{\theta}{2})|0\rangle-e^{i\phi}\cos(\frac{\theta}{2})|1\rangle$. From the basis of  $\sigma_{3}(\theta, \phi)$, we can now form the other mutually unbiased bases such as 
\begin{eqnarray}
|x\pm(\theta, \phi)\rangle&=&\frac{|0(\theta, \phi)\rangle\pm|1(\theta, \phi)\rangle}{\sqrt{2}}\nonumber\\
|y\pm(\theta, \phi)\rangle&=&\frac{|0(\theta, \phi)\rangle\pm i|1(\theta, \phi)\rangle}{\sqrt{2}}
\end{eqnarray}
In this article, we use these MUBs to perform coherence measurements and optimize over $(\theta, \phi)$ to get the maximum possible violation of the inequalities.
\subsection{B. Steering complementarity relationships}We now derive a set of steering complementarity relations for various steering inequalities for the case of definite causal order. These relations are complementary in the sense that if one of the steering inequalities is violated by a state, its complementary part in the complementarity relation does not violate the corresponding steering inequality.

We consider the quantity $S_{(i,j,k)}$, defined as follows, and show that for any arbitrary quantum state, it is bounded. Since the analysis holds for arbitrary quantum states, the bound cannot be violated by quantum theory.
\begin{myeqnarrayB} 
S_{(i,j,k)}&=&\sum^1_{a,b=0}\sum_{i, j, k}p(a|i)p(b|j)C^{\alpha}_{k}(\rho'_{a|i})C^{\alpha}_{k}(\rho'_{b|j})\nonumber\\&\leq &\frac{1}{2}\sum^1_{a,b=0}\sum_{i, j, k} p(a|i)p(b|j)\left(C^{\alpha^2}_k(\rho'_{a|i})+C^{\alpha^2}_k(\rho'_{b|j})\right)\nonumber\\&\leq&
\sum^1_{a,b=0}\sum_{i, j}p(a|i) p(b|j)\Omega_{\alpha}\nonumber
\\&\leq&
9 \Omega_{\alpha}.
\label{proof3}
\end{myeqnarrayB}

To arrive at the bound in the above Eq. (\ref{proof3}), we use the following facts in the first and second inequality respectively: i). For any two real numbers $x$ and $y$, $xy\leq\frac{x^2+y^2}{2}$ and ii). Coherence complementarity relations from the Eq. (5). 

The quantity $S_{(i,j,k)}$ can be decomposed into several parts for which a set of steering inequalities  can be derived. For example, we consider the following decomposition,
\begin{myequationB}
S_{(i,j,k)}=S_{(i=j,k)}+S_{(i=k\neq j)}+S_{(i\neq j=k)}+S_{(i\neq j\neq k)},
\label{decomposition}
\end{myequationB}
where, the subscript for each term denotes the choice of basis $i$, $j$ and $k$. Below, we show that each term represents a steering inequality and for each term, a steering inequality can be derived. A new steering inequality can be derived by add two or more terms in the decomposition but not all of them together.
\begin{figure}[h!]
\includegraphics[scale=1]{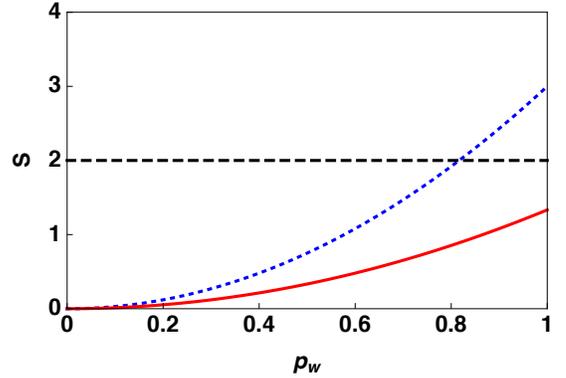}
\caption{We plot $\frac{S_{(i\neq j\neq k)}}{2}$ (Blue, dotted) and $\frac{S_{(i,j,k)}}{9}$ (Red, solid curve) from Eq. (\ref{sdecomp3}) and Eq. (\ref{proof3}) respectively with $p_w$ for Werner state~(15) and $l_1$-norm measure of quantum coherence. Bounds for both of the quantities turn out to be $\Omega_\alpha$. We observe that whereas $\frac{S_{(i\neq j\neq k)}}{2}$ violates the bound for around $p_w= 0.816$, the quantity $\frac{S_{(i,j,k)}}{9}$ does not show the violation for any value of $p_w$.}
\label{sijk_l1}
\end{figure}

In the paper, we have explicitly derived the bound of the quantity $S_{(i\neq j\neq k)}$, while the bound for the rest of the quantities can be derived following the same method. Here we elucidate the various decompositions of $S_{(i,j,k)}$ and derive their bounds following a new method. These bounds are relatively weaker than those found following the method given in the paper. However, one advantage of these new bounds is that sum of these bounds of all the terms in the decomposition gives the bound given by the Eq. (\ref{proof3}). We explicitly calculate the new bound for the quantity $S_{(i=j,k)}$ below, while bounds for the rest of the quantities can be derived following the same method as shown below.

\begin{myeqnarrayB}
&&\sum^1_{a,b=0}\sum_{i=j, k}p(a|i)p(b|i)C^{\alpha}_{k}(\rho'_{a|i})C^{\alpha}_{k}(\rho'_{b|i})\nonumber\\&&\leq \frac{1}{2}\underset{\lambda,\lambda'}{\sum^1_{a,b=0}\sum_{i=j, k}} P_{\lambda}P_{\lambda'}p(a|i,\lambda)p(b|i,\lambda')\nonumber\\ &&\left( C^{\alpha^2}_k(\rho'_{\lambda}) + C^{\alpha^2}_k(\rho'_{\lambda'})\right)\nonumber\\&&\leq 
\underset{\lambda,\lambda'}{\sum^1_{a,b=0}\sum_{i=j}}P_{\lambda}P_{\lambda'}p(a|i,\lambda) p(b|i,\lambda')\Omega_{\alpha}
\leq
3 \Omega_{\alpha}.
\label{proof4}
\end{myeqnarrayB}

\begin{figure}[h!]
\includegraphics[scale=1]{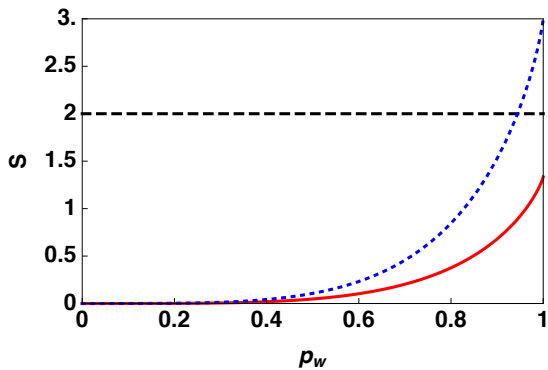}
\caption{We plot $\frac{S_{(i\neq j\neq k)}}{2}$ (Blue, dotted) and $\frac{S_{(i,j,k)}}{9}$ (Red, solid curve) from Eq. (\ref{sdecomp3}) and Eq. (\ref{proof3}) respectively with $p_w$ for Werner state~(15) and relative entropy measure of quantum coherence. Bounds for both of the quantities turn out to be $\Omega_\alpha$. We observe that whereas $\frac{S_{(i\neq j\neq k)}}{2}$ violates the bound for $p_w=0.944$, the quantity $\frac{S_{(i,j,k)}}{9}$ does not show the violation for any value of $p_w$.}
\label{sijk_rec}
\end{figure}

In a similar manner, we obtain,

\begin{myequationB}
S_{(i\neq j=k)}\leq 2 \Omega_{\alpha},
\label{sdecomp1}
\end{myequationB}

\begin{myequationB}
S_{(i=k\neq j)}\leq 2 \Omega_{\alpha},
\label{sdecomp2}
\end{myequationB}

\begin{myequationB}
S_{(i\neq j\neq k)}\leq 2 \Omega_{\alpha}.
\label{sdecomp3}
\end{myequationB}

We plot the behaviour of $S_{(i,j,k)}$ and $S_{(i\neq j\neq k)}$ for varying $p_W$ using $l_1$ norm and relative entropy of coherence in Fig. (\ref{sijk_l1}) and Fig. (\ref{sijk_rec}) respectively. It is observed that while $S_{(i\neq j\neq k)}$ shows the violation for both, $S_{(i,j,k)}$ is always satisfied.

 If a steering inequality corresponding to a term or a particular group of terms appearing in the decomposition in Eq.~(\ref{decomposition}) is violated  for a particular state, the rest of the terms together cannot violate the corresponding steering inequality and in fact compensate for the violation of the former inequality such that $S_{(i,j,k)}$ is always satisfied.

\subsection{D. Tightness of the bound under 1SQI}
We would like to find at least one example of a (set of) quantum state(s) under the paradigm of 1SQI-model, which 
saturates the bound of $3\Omega_{\alpha}$. This will ensure that the bound is tight and no further improvement can be made. Without loss of generality, we assume that the strategy to 
determine the conditional states $\rho'_{a|k}$, which depend on the outcome $a$, $\lambda$ and choice of basis $k$ as follows.

For each choice of basis $k$ in which coherence will be measured, choose $\rho'_{a, \lambda}$ and 
$\rho'_{b, \lambda'}$ to be a pure state which is diagonal in the $i\neq k$ basis.

Let us evaluate the quantity $S$ defined as in the main text for the term $k=1$,
while choosing $\rho'_{a, \lambda} = \rho'_{b,\lambda'}=|0\rangle$ $\forall (a, b)$. Then we have,

\begin{myeqnarrayD}
S_{k=1}&=&\underset{\lambda, \lambda'}{\sum_{a,b=0}^{1}}\sum_{i\neq j=2}^{3}P_\lambda P_{\lambda'}p(a|i,\lambda)
p(b|j,\lambda') C_k^\alpha(\rho'_{a,\lambda})C_k^\alpha(\rho'_{b,\lambda'})\nonumber\\
&=&2\left(\frac{1}{2}+0+\frac{1}{2}+0\right) = 2,
\end{myeqnarrayD}
where for the last equality we have expanded the summation over $a$ and $b$ and used the fact that 
the cases $i=2$, $j=3$ and $i=3$, $j=2$ are symmetric 
and will yield the same results.

A similar analysis for $k=2$ and $k=3$ under the aforementioned strategy for selecting 1SQI states also yields the same 
value. Therefore for at least one set of states we have shown that 
\begin{myequationD}
S=S_{k=1}+S_{k=2}+S_{k=3} = 6,
\end{myequationD}
which completes the required proof. 

It should be noted that a similar prescription for ordered causal models does not yield value of $S$ higher than
$2\Omega_{\alpha}$.
\subsection{E. Generalization of LHS bound}
In this section, we generalize the LHS bound on average local quantum coherence in the $d$-dimension, where $d$ is prime power. We know that the number of MUBs in prime power dimension $d$ is $d+1$. The coherence complementarity relation (considering the definition of coherence to be normalized by the factor of $\frac{1}{d-1}$) for $l_1$-norm turns out to be \cite{hall}
\begin{myequationE}
\sum_{j}C_{j}^{l_1^2}(\rho)\leq \frac{d}{(d-1)}(dP(\rho)-1)\leq d,
\label{ccgen}
\end{myequationE}
where $P(\rho)=\text{Tr}(\rho^2)\leq 1$ is the purity of the state. This inequality can now be used to derive the LHS bound on the quantity for bipartite states with subsystems of prime-power dimensions($d$). To find the bound, we show
\begin{widetext}
\begin{myeqnarrayE}
S&=&\sum^{d-1}_{a,b=0}\sum_{\substack{i, j, k = 1\\ i\neq j\neq k}}^{d+1}p(a|i)p(b|i)C^{l_1}_{k}(\rho'_{a|i})C^{l_1}_{k}(\rho'_{b|j})\overset{LHS}{\leq} \frac{1}{2}{\sum^{d-1}_{\substack{a,b=0\\ \lambda,\lambda'}}\sum_{\substack{i, j, k = 1\\ i\neq j\neq k}}^{d+1}} P_{\lambda}P_{\lambda'}p(a|i,\lambda)p(b|j,\lambda')\left(C^{l_1^2}_k(\rho'_{\lambda})+C^{l_1^2}_k(\rho'_{\lambda'})\right)\nonumber\\&=&\frac{1}{2}{\sum^{d-1}_{\substack{a,b=0\\ \lambda,\lambda'}}\sum_{k=1}^{d+1}\sum_{\substack{m,n=1,\\ m> n}}^{d}}P_{\lambda}P_{\lambda'}\bigg{[}p(a|\mathcal{K}_m^k,\lambda)p(b|\mathcal{K}_n^k,\lambda')+p(b|\mathcal{K}_m^k,\lambda')p(a|\mathcal{K}_n^k,\lambda)\bigg{]}\left(C^{l_1^2}_k(\rho'_{\lambda})+C^{l_1^2}_k(\rho'_{\lambda'})\right)\nonumber\\&\leq &
\frac{1}{2}{\sum_{\substack{\lambda,\lambda'}}\sum_{k=1}^{d+1}}d(d-1)P_{\lambda}P_{\lambda'}\left(C^{l_1^2}_k(\rho'_{\lambda})+C^{l_1^2}_k(\rho'_{\lambda'})\right)
\leq d^2(d-1)\sum_{\lambda,\lambda'}P_{\lambda}P_{\lambda'}
= d^2(d-1),
\label{lhsgen}
\end{myeqnarrayE}
\end{widetext}
where $\mathcal{K}^{k}_{n}=Modulo(k+n-1,d+1)+1$. Here, if the $a$'s and $b$'s in the second line is summed up, each term in the square bracket gives the value one. Similarly, if $m$'s and $n$'s are summed, there will be $d(d-1)/2$ such pairs of terms inside square bracket giving raise to the $d(d-1)$ term in the third line. In the last inequality, we use the coherence complementarity relation as given in Eq. (\ref{ccgen}).

\subsection{F. Generalization of 1SQI bound}
One can similarly, find the bound on the sum of local quantum coherence-squares in the MUBs for 1SQI model for bipartite states with subsystems of prime power dimensions ($d$). As before, we start with the quantity $S$ as 
\begin{widetext}
\begin{myeqnarrayF}
S & = &\underset{\lambda,\lambda'}{\sum_{i\neq j\neq k}^{d+1}\sum_{a,b =0}^{d-1}}P_\lambda P_{\lambda'}p(a|i,\lambda)p(b|j,\lambda')
C_k^{\alpha}(\rho_{a,\lambda})C_k^{\alpha}(\rho_{b,\lambda'})\nonumber\\
& \leq &\frac{d-1}{2}\underset{\lambda,\lambda', i, k = 1}{\sum_{i\neq k}^{d+1}}P_\lambda P_{\lambda'}
\left[\sum_{a,b=0}^{d-1} \bigg{(}p(a|i,\lambda)C_k^{\alpha}(\rho_{b,\lambda'})\bigg{)}^2+
\sum_{a,b=0}^{d-1}\bigg{(}p(b|j,\lambda)C_k^{\alpha}(\rho_{a,\lambda})\bigg{)}^2\right]\nonumber\\ & = & S_1+S_2,
\end{myeqnarrayF}
\end{widetext}
where, we consider, the first term  $S_1$ to be the first term of the second line in the above equation and $S_2$ to be the second. In the above inequality,  we use the fact that for any positive real numbers $p_{i}$, $q_i$, $r_i$ and $s_i$,
\begin{myequationF}
\underset{i, j, k = 1}{\sum_{i\neq j\neq k}^{d+1}} p_i q_j r_k s_k\leq
\frac{d-1}{2}\left[\underset{i, j = 1}{\sum_{i\neq j}^{d+1}} \left( p_i r_j \right)^2+\underset{i, j = 1}{\sum_{i\neq j}^{d+1}}\left(q_i s_j\right)^2\right].
\end{myequationF}
Now, we can easily expand over the $i-$index and show that 
\begin{myeqnarrayF}
S_1&=&\frac{d-1}{2}\underset{\lambda,\lambda', i, k = 1}{\sum_{i\neq k}^{d+1}}P_\lambda P_{\lambda'}
\sum_{a,b=0}^{d-1} \bigg{(}p(a|i,\lambda)C_k^{\alpha}(\rho_{b,\lambda'})\bigg{)}^2\nonumber\\ &=&\frac{d-1}{2}\underset{\lambda,\lambda'}{\sum_{k = 1}^{d+1}}P_\lambda P_{\lambda'}
\sum_{a,b=0}^{d-1} \bigg{(}F(a|k,\lambda)C_k^{\alpha}(\rho_{b,\lambda'})\bigg{)}^2,
\end{myeqnarrayF}
where $F(a|k, \lambda)=\underset{i\neq k}{\sum_{i=1}^{d+1}}p^2(a|i, \lambda)$. One can easily show that $\sum_{a=0}^{d-1}F(a|k, \lambda)\leq \frac{d+1}{2} \hspace{0.1cm}\forall k$.   This implies that $S_1\leq \frac{d(d^2-1)}{2}$ and similarly, one can show that $S_2$ as well bounded from above by $\frac{d(d^2-1)}{2}$, which results, 
\begin{myequationF}
S\leq d(d^2-1)
\end{myequationF}

\end{document}